\documentclass{article}

\usepackage{arxiv}

\usepackage[utf8]{inputenc} 
\usepackage[T1]{fontenc}    
\usepackage{hyperref}       
\usepackage{url}            
\usepackage{booktabs}       
\usepackage{amsfonts}       
\usepackage{nicefrac}       
\usepackage{microtype}      
\usepackage{lipsum}		
\usepackage{graphicx}

\usepackage{amsmath}
\usepackage{bm}
\usepackage{bbm}
\usepackage{MnSymbol}
\usepackage{natbib}
\usepackage{enumitem}
\usepackage{graphicx}
\usepackage{soul}
\usepackage{xcolor}
\usepackage{url}
\usepackage[ruled,vlined]{algorithm2e}
\usepackage[normalem]{ulem}
\usepackage{todonotes}

\usepackage{amsthm}

\newtheorem{theorem}{Theorem}
\newtheorem{proposition}{Proposition}
\newtheorem{lemma}{Lemma}
\newtheorem{corollary}{Corollary}

\theoremstyle{definition}
\newtheorem{definition}{Definition}
\newtheorem{assumption}{Assumption}

\theoremstyle{remark}
\newtheorem*{remark}{Remark}

\newcommand{\T}{{ \mathrm{\scriptscriptstyle T} }}
\newcommand{\dd}{\,\mathrm{d}}
\newcommand{\RR}{\mathbbm{R}}
\newcommand{\law}{\mathbbm{P}}
\newcommand{\GcondLaw}{\mathbbm{G}}
\newcommand{\condLawFull}{\mathbbm{P}^{(T,x_0,x_T)}}
\newcommand{\condLaw}{\mathbbm{P}^{\star}}
\newcommand{\knots}{\mathcal{K}}
\newcommand{\anchors}{\mathcal{A}}

\newcommand{\blockG}{\mathcal{G}}
\newcommand{\Ltwo}{\mathcal{L}^2}

\DeclareMathOperator{\sech}{sech}

\newcommand{\cta}{\ensuremath{c_2}}
\newcommand{\ctb}{\ensuremath{c_3}}
\newcommand{\ctc}{\ensuremath{c_1}}
\newcommand{\ctd}{\ensuremath{c_5}}
\newcommand{\cte}{\ensuremath{c_4}}
\newcommand{\ctf}{\ensuremath{c_8}}
\newcommand{\ctg}{\ensuremath{c_6}}
\newcommand{\cth}{\ensuremath{c_7}}

\title{The computational cost of blocking for sampling discretely observed diffusions}

\author{ \href{https://orcid.org/0000-0003-3031-8290}{\includegraphics[scale=0.06]{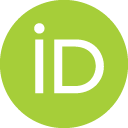}\hspace{1mm}Marcin Mider}
	\\
	Max Planck Institute for Mathematics in the Sciences\\
	Inselstraße 22, 04103 Leipzig, Germany\\
	\texttt{marcin.mider@mis.mpg.de} \\
	\And
	 \href{https://orcid.org/0000-0001-7603-4205}{\includegraphics[scale=0.06]{orcid.png}\hspace{1mm}Paul A.~Jenkins} \\
	Department of Statistics\\
	University of Warwick\\
	Coventry CV4 7AL, U.K.\\
	\texttt{p.jenkins@warwick.ac.uk} \\
	\And
	Murray Pollock \\
	School of Mathematics, Statistics and Physics\\
	Newcastle University\\
	Newcastle-upon-Tyne NE1 7RU, U.K.\\
	\texttt{murray.pollock@newcastle.ac.uk} \\
	\And
	Gareth O.~Roberts \\
	Department of Statistics\\
	University of Warwick\\
	Coventry CV4 7AL, U.K.\\
	\texttt{gareth.o.roberts@warwick.ac.uk} \\
}

\date{}

\renewcommand{\headeright}{}
\renewcommand{\undertitle}{}
\renewcommand{\shorttitle}{Polynomial blocking strategies}

\hypersetup{
pdftitle={The computational cost of blocking for sampling discretely observed diffusions},
pdfsubject={60J10, 60J22, 60J60, 60J65},
pdfauthor={Marcin Mider, Paul A.~Jenkins, Murray Pollock, Gareth O.~Roberts},
pdfkeywords={Bayesian inference, Blocking, Diffusion, Gaussian process, Markov chain Monte Carlo},
}

\begin{document}

\maketitle

\begin{abstract}
Many approaches for conducting Bayesian inference on discretely observed diffusions involve imputing diffusion bridges between observations. This can be computationally challenging in settings in which the temporal horizon between subsequent observations is large, due to the poor scaling of algorithms for simulating bridges as observation distance increases. It is common in practical settings to use a \emph{blocking scheme}, in which the path is split into a (user-specified) number of overlapping segments and a Gibbs sampler is employed to update segments in turn. Substituting the independent simulation of diffusion bridges for one obtained using blocking introduces an inherent trade-off: we are now imputing shorter bridges at the cost of introducing a dependency between subsequent iterations of the bridge sampler. This is further complicated by the fact that there are a number of possible ways to implement the blocking scheme, each of which introduces a different dependency structure between iterations. Although blocking schemes have had considerable \emph{empirical} success in practice, there has been no analysis of this trade-off nor guidance to practitioners on the particular specifications that should be used to obtain a computationally efficient implementation. In this article we conduct this analysis and demonstrate that the expected computational cost of a blocked path-space rejection sampler applied to Brownian bridges scales asymptotically at a cubic rate with respect to the observation distance and that this rate is linear in the case of the Ornstein--Uhlenbeck process. Numerical experiments suggest applicability both of the results of our paper and of the guidance we provide beyond the class of linear diffusions considered.

\keywords{Bayesian inference \and Blocking \and Diffusion \and  Gaussian process \and  Markov chain Monte Carlo}
\end{abstract}


\section{Introduction} \label{sec:intro}
Diffusions have been widely applied to model continuous-time phenomena of interest, including molecular dynamics \citep{boys2008bayesian}, neuroscience \citep{lansky2008review}, and finance \citep{karatzas1998methods}. In general, a diffusion on $\RR^d$ is a Markov process $X$ defined to be the solution, with law we will denote by $\law$, to a stochastic differential equation of the following form:
\begin{align}
  \dd X_t &= b(X_t) \dd t + \sigma(X_t) \dd W_t,\quad X_0=x_0,\quad t\in[0,T], \label{eq:master_sde}
\end{align}
where $b:\RR^d\to\RR^d$ and $\sigma:\RR^d\to\RR^{d\times d'}$ denote the drift and volatility coefficient respectively, and $W$ is a standard $d'$-dimensional Brownian motion. Throughout we assume standard regularity conditions hold which ensure the existence of a unique, global, weak solution to \eqref{eq:master_sde} \citep[see for instance][]{oksendal2003stochastic}.

In practice we will typically only have access to discrete observations of \eqref{eq:master_sde}, and so for practitioners the statistical problem of interest is to use these observations to draw inference on the parameters of $b$ and $\sigma$ of \eqref{eq:master_sde}. A common Bayesian strategy is to augment the parameter space with the space describing the complete underlying diffusion trajectory. A Markov Chain Monte Carlo algorithm can then explore this augmented space by alternating between updates of the parameters and updates of the unobserved sample path connecting observations (sampling of diffusion bridges) \citep{roberts2001inference}. As a consequence a considerable and methodologically diverse literature has been developed concerned with simulating diffusion bridges (the law of \eqref{eq:master_sde} conditioned to terminate at the subsequent observation (for instance, $X_T=x_T$), which we denote by $\condLawFull$ or generically by $\condLaw$), including  \citet{beskos2005exact,bladt2014simple,delyon2006simulation,durham2002numerical,golightly_bayesian:2008,hairer2011sampling,roberts2001inference,schauer2017guided}. 
 
One of the common difficulties with Markov Chain Monte Carlo strategies is sampling diffusion bridges between distant observations; the \emph{duration} of the bridge, which we denote by $T$, is large. This setting naturally arises when the underlying diffusion \eqref{eq:master_sde} is \emph{sparsely} observed (or high-dimensional), for instance in shape analysis applications \citep{arnaudon2020diffusion}, or in the case of diffusions on graphs \citep{freidlin1993diffusion}. The problem here is that methodologies for sampling diffusion bridges scale poorly with $T$, and many of the most widely used approaches have \emph{exponential computational cost} in $T$. Consequently, addressing the poor scaling in $T$ has drawn considerable interest. One popular approach is the \emph{blocking} scheme introduced by \cite{shephard1997likelihood}, which has been employed in a number of practical problems with strong empirical evidence of its efficacy \citep{chib2004likelihood, golightly_bayesian:2008, kalogeropoulos2007likelihood, kalogeropoulos2010inference,van2018bayesian,stramer2007bayesian}.

Blocking is a conceptually simple idea in which the time domain of the diffusion bridge is overlaid with a set of \emph{temporal anchors} ($0=:k_0<k_1<\dots<k_m<k_{m+1}:=T$), and the values of the bridge are taken for some \emph{initialisation} trajectory at those points (which are known as \emph{knots}, and for which we will denote $X_i:=X_{k_i}$ to simplify notation). Simulation from $\condLawFull$ is then achieved by constructing a Gibbs sampler which alternates between updating knots and updating the segments of the trajectory conditional on the knots, a number of times. For instance, we could begin by simulating from the conditional law $\mathbbm{P}^{(k_2-k_0,X_0,X_{2})}$ (updating the trajectory between $[t_0,t_2]$ which includes the knot at $X_1$, conditional on the knots at $X_0$ and $X_2$), and then $\mathbbm{P}^{(k_3-k_1,X_1,X_{3})}$ (updating the trajectory between $[t_1,t_3]$ and containing the knot $X_2$, conditional on the knots at $X_1$ and $X_3$), and so on \emph{sweeping} across all anchor points. This sweep would then be iterated a number of times to reduce the dependency between the resulting bridge and the initial (or previous) trajectory. In this article we consider the three canonical blocking schemes of \citet{roberts1997updating} with equidistant anchors: the \emph{checkerboard} scheme, in which the odd and even indexed knots are alternatively updated; the \emph{lexicographic} scheme, in which the knots are updated in temporal order; and the \emph{random} scheme, in which at each step a random knot is updated. We more formally introduce blocking and define these schemes in Section \ref{sec:blocking}.

From a computational perspective, blocking substitutes the expensive simulation of a (single independent) draw from $\condLawFull$, with the cost of simulating repeated sweeps of the $m+1$ shorter (and computationally more efficient) bridges for each segment given by the temporal anchors. Any analysis of this trade-off needs to take into account the serial correlation induced by the blocking strategy.

Despite widespread adoption of blocking in practice to mitigate the computational cost of simulating diffusion bridges (as indicated above), there is little theoretical support for its efficacy. Furthermore, there is little concrete guidance on how to implement, and then appropriately tune (selecting for instance the number and locations of the anchor points), a blocking scheme. 

In this article we provide general guidance for implementing blocking schemes by addressing these practical considerations for particular classes of diffusion process. We analyse the computational cost of several rejection sampling algorithms for bridges as a function of block size and bridge duration. In all cases we consider a fixed regular spacing of $m$ anchor points as $m,T\to\infty$, in contrast with the study of the \emph{`in-fill'} asymptotic of \citet{roberts2001inference} in which $T$ is fixed and $m\to\infty$. We analyse the expected cost of a single iteration of various algorithms, and then to capture the trade-off described above we consider the cost of the algorithm which comprises both the cost of one iteration, and the total number of iterations required to obtain an `independent' sample. We give a more formal description of what we mean by achieving independence below, in terms of the relaxation time of the underlying Markov chain.

In this article we work under the assumption that the underlying measure is a Gaussian diffusion (i.e.\ $\law$ is the law of a scaled Brownian motion or the law of the Ornstein--Uhlenbeck process). Under this simplification the Gibbs step for updating the bridge segments can be implemented without error, i.e.\ without discretising time, for example by means of a rejection sampler directly on the path-space of the diffusion (see Appendix \ref{sec:psrs} for full details). In this setting we prove that Theorem \ref{thm:main_thm_comp_cost} below holds, as the culmination of the results in Section \ref{sec:computational_cost_of_blocking}. We gather all proofs in the appendices.

\begin{theorem}\label{thm:main_thm_comp_cost}
  Suppose $\condLaw$ is the conditional law of a Gaussian diffusion which is sampled by rejection on path-space and using a checkerboard or lexicographic or random blocking scheme. Suppose the $m$ anchors are spaced equidistantly such that $m=\ctc T$ (for some constant $\ctc >0$). Then the expected computational cost of the blocked rejection sampler, $C_{\texttt{blocking}}(T)$, satisfies:
  \begin{equation}\label{eq:cost_of_blocking_explicit}
    C_{\texttt{blocking}}(T)=\mathcal{O}(T^3),\qquad\mbox{as}\quad T\to\infty,
  \end{equation}
  whenever $\law$ denotes the law of a scaled Brownian motion and
  \begin{equation}\label{eq:cost_of_blocking_explicit_OU}
    C_{\texttt{blocking}}(T)=\mathcal{O}(T),\qquad\mbox{as}\quad T\to\infty,
  \end{equation}
  whenever $\law$ denotes the law of the Ornstein--Uhlenbeck process.
\end{theorem}

\begin{remark} \label{rem:oubetter}
 Note that in the case of a Brownian bridge there is long range dependency in the path, in the sense that the correlation between $X_s$ and $X_t$ is non-negligible even for $0 \ll s \ll t \ll T$. On the other hand, for the Ornstein--Uhlenbeck process its ergodicity breaks this dependency. For an Ornstein--Uhlenbeck process whose drift is of the form $b(X_t) = -\theta X_t$ and for $T\gg 0$, there is a phase transition in its behaviour as $\theta \to 0$ in that the computational cost of a blocked rejection sampler for Brownian motion is \emph{not} recovered. Recall that in this paper we are working under the assumption that the underlying measure is a Gaussian diffusion, but in most practical settings the \emph{target law} will be more complicated. In such settings it would be typical to use a Gaussian diffusion as a \emph{proposal law} for the non-Gaussian target law. In principle Theorem \ref{thm:main_thm_comp_cost} would suggest that an Ornstein--Uhlenbeck process proposal for a \emph{stationary} target law would be advantageous over a Brownian bridge proposal, although in practice this predicted computational saving would depend on how closely the target process matched the invariant distribution of the Ornstein--Uhlenbeck process.
\end{remark}

Theorem \ref{thm:main_thm_comp_cost} contrasts sharply with the case without blocking. We show later in Proposition \ref{prop:comp_cost} that, for a $d$-dimensional Brownian bridge proposal in the absence of blocking, the cost is exponential in $T$. Although what we prove in Theorem \ref{thm:main_thm_comp_cost} addresses a somewhat idealised setting, the requirement $m=\ctc T$ acts as a concrete guide for choosing the number of blocks. Furthermore, our empirical results in Section \ref{sec:numerical} indicate that the guidance we establish can be more broadly useful beyond the class of linear diffusions. Thus we demonstrate that blocking can lead to significantly improved computational efficiency when conducting inference for discretely observed diffusions.


\section{Blocking} \label{sec:blocking}
In this section we provide a systematic definition of blocking for sampling a diffusion path. Define a set of anchors spread across the time domain: $0<k_1<\dots<k_m<T$ and \emph{knots}  as the values of the path taken at the anchors:
\begin{equation*}
    \knots(\omega):=\{ X_{k_1}(\omega), \dots,X_{k_m}(\omega) \}.
\end{equation*}
Each anchor is now assigned to one of $\mathbbm{k}$ disjoint subsets $\anchors_i$, $i=1,\dots,\mathbbm{k}$, each comprising $m_i$ anchors:
\begin{equation*}
    \{k_1,\dots,k_m\}=\bigcupdot_{i=1}^\mathbbm{k}\{r_{i1},\dots,r_{im_i}\}=:\bigcupdot_{i=1}^\mathbbm{k}\anchors_i,\quad (\mbox{with } m_i\in\mathbbm{N}_+,\, i\in\{1,\dots,\mathbbm{k}\}).
\end{equation*}
In particular, $\mathcal{A}_i = \{r_{i1},\dots,r_{im_i}\}$ is the set of anchors associated (uniquely) to the $i$th block. This allows us to group the knots by associating them with the corresponding subsets of anchors:
\begin{equation*}
    \knots_i(\omega):=\{X_r(\omega)\,;\,r\in\anchors_i\},\quad i\in\{1,\dots,\mathbbm{k}\}.
\end{equation*}
For convenience of notation we let $\knots_{-i}$ (resp.\ $\anchors_{-i}$) denote all knots (resp.\ anchors) that do not belong to $\knots_i$ (resp.\ $\anchors_i$):
\begin{align*}
    \knots_{-i}(\omega)&:=\bigcup_{j\neq i}\knots_j(\omega), 
    & \anchors_{-i}(\omega)&:=\bigcupdot_{j\neq i}\anchors_j(\omega),& i&\in\{1,\dots,\mathbbm{k}\},
\end{align*}
and assign labels to an ordered collection of all anchors in $\anchors_{-i}$, plus the end-points:
\begin{equation*}
    \{e_{i0},\dots,e_{i(m+1-m_i)}\} = \anchors_{-i} \cup \{0, T\},\quad (\mbox{with } e_{ij}<e_{i(j+1)}),\quad i\in\{1,\dots,\mathbbm{k}\}.
\end{equation*}
Further, define
\begin{equation*}
    \mathcal{B}_i:=\left\{\left.(e_{ij},e_{i(j+1)})\,\right|\,\exists r\in\anchors_i\mbox{ s.t. }r\in[e_{ij},e_{i(j+1)}]\right\}_{j=0}^{m-m_i},
\end{equation*}
to be only those intervals between the end-points or anchors in $\anchors_{-i}$, which contain at least one anchor belonging to $\anchors_i$. The path segments $X|_{\mathcal{B}_i}$, obtained through restricting $X$ to $\mathcal{B}_i$, are termed \emph{blocks}. Finally, in the case $\mathbbm{k}=2$ we say that $\anchors_1$ and $\anchors_2$ are \emph{interlaced} if whenever $a,c\in\anchors_i$, with $a<c$, then there exists $b\in\anchors_{(i\mod 2)+1}$ s.t. $a<b<c$, $i=1,2$.

A sampler for a path equipped with a blocking technique is a Gibbs sampler that updates the full path only one block at a time by drawing from the conditional laws $\condLaw|_{\mathcal{B}_i}(\cdot | \knots_{-i})$---i.e. the target laws restricted to blocks $\mathcal{B}_i$ and conditioned on the knots in $\knots_{-i}$. For simplicity we refer to this technique as a \emph{blocked sampler} in the remainder of the text, and present general pseudo-code for it in Algorithm \ref{alg:blocking}.

\begin{algorithm}[H]
    \caption{Blocked sampler on path-space}\label{alg:blocking}
        Initialise $X$\;
        \For{$n=1,\dots,N$}{
            \For{$i=1,\dots,\mathbbm{k}$}{
                Draw $I\sim q(i,\cdot)$ (various choices for $q$ are defined below, in Definitions \ref{def:cbus}--\ref{def:rbus})\;
                Update $X|_{\mathcal{B}_I}$ by sampling $X|_{\mathcal{B}_I}\sim\condLaw|_{\mathcal{B}_I}(\cdot|\mathcal{K}_{-I})$\;
            }
        }
        \Return{$X$}
\end{algorithm}

There are a number of ways we can update the blocks, and in this article we consider the three canonical blocking schemes of \citet{roberts1997updating}. In particular, we refer to a single, full Gibbs sweep of  Algorithm \ref{alg:blocking} (the inner \texttt{for-loop}) as a:

\begin{definition}\label{def:cbus}
    \emph{Checkerboard} blocking update scheme if $\mathbbm{k}=2$, $\anchors_1$ and $\anchors_2$ are interlaced, and $q(i,j):=\mathbbm{1}_{\{i\}}(j)$.
\end{definition}

\begin{definition}\label{def:lbus}
    \emph{Lexicographic} blocking update scheme if $\mathbbm{k}=m$, $\anchors_i:=\{k_i\}$, $i\in\{1,\dots,m\}$, and $q(i,j):=\mathbbm{1}_{\{i\}}(j)$.
\end{definition}

\begin{definition}\label{def:rbus}
    \emph{Random} blocking update scheme if $\mathbbm{k}=m$, $\anchors_i:=\{k_i\}$, $i\in\{1,\dots,m\}$, and $q(i,j):=\frac{1}{m}\mathbbm{1}_{\{1,\dots,m\}}(j)$.
\end{definition}

The above are not exhaustive, but characterise the most widely used, and are tractable enough for analysis. We further simplify various computations by assuming the anchors are \emph{equidistant}, and defer discussion of this assumption and its relaxation to Section \ref{sec:discussion}.

\begin{assumption}\label{as:equidistant_grid}
  The anchors are placed on an equidistant grid:
    \[
        k_{i+1}-k_i=\frac{T}{m+1}=:\delta_{m,T},\quad i\in\{1,\dots,m-1\}.
    \]
\end{assumption}

As mentioned in the Introduction, we will study the asymptotic regime in which $\delta_{m,T}$ is fixed as $m,T\to\infty$.


\section{Computational analysis}\label{sec:computational_cost_of_blocking}


\subsection{Cost of a single sweep}

We begin by quantifying the computational cost of a rejection sampling algorithm for diffusion bridges in the absence of blocking. The setting here, as given in further detail in Appendix \ref{sec:psrs}, is one in which the target law is absolutely continuous with respect to a $d$-dimensional Brownian bridge proposal path.

\begin{proposition}\label{prop:comp_cost}
  Under Assumptions 3--9 enumerated in Appendix \ref{sec:psrs}, the expected computational cost as a function of\/ $T$ of obtaining a single draw with a path-space rejection sampling algorithm, denoted by $C_{\texttt{rej}}(T)$, is given by
  \begin{equation}\label{eq:comp_cost_psrs}
    C_{\texttt{rej}}(T) = f(T)Te^{\cta T},
  \end{equation}
  where $\cta >0$ is some constant independent of $T$, and the function $f:\RR_+\to\RR$ is continuous and such that $f(T)\sim T^{-d/2}$ as $T\to\infty$.
  In particular, for large enough $T$ there is a constant $\ctb >0$ such that:
  \begin{equation*}
    C_{\texttt{rej}}(T) \geq \ctb T^{1-d/2}e^{\cta T}.
  \end{equation*}
\end{proposition}
\begin{remark}
Note that Proposition \ref{prop:comp_cost} does not stipulate in what way the constant $\ctb$ might vary with dimension. Without further structure it is impossible to characterise this behaviour. However it is highly likely that   $\ctb$ will increase at least linearly with dimension (as for instance would be the case for diffusions consisting of $d$ independent components). 
\end{remark}
\begin{remark}
    If $\law$ is the law of a drifted Brownian motion, then Proposition \ref{prop:comp_cost} cannot be applied directly, because Assumption \ref{as:ergodicity} does not hold. However, for this case an easy calculation shows that the acceptance probability of a rejection sampler with Brownian bridge proposals is equal to $1$, implying (under Assumption \ref{as:x_cost}) that $C_{\texttt{rej}}(T)$ is proportional to $T$.
\end{remark}

Now considering a single sweep of the blocking schemes introduced in Section \ref{sec:blocking}, note that we have substituted sampling a single diffusion bridge (of length $T$) with sampling a number of diffusion bridges of shorter time horizon, $2\delta_{m,T}$ (for example, to ensure the point $X_{2\delta_{m,T}}$ is updated one could sample a new bridge of length $2\delta_{m,T}$ connecting $X_{\delta_{m,T}}$ with $X_{3\delta_{m,T}}$). By application of Proposition \ref{prop:comp_cost}, the expected computational cost of simulating each of these shorter bridges is therefore $C_{\texttt{rej}}(2\delta_{m,T})$, and hence the expected cost of a single Gibbs sweep is:
\begin{equation}\label{eq:gibbs_sweep_cost1}
  C_{\texttt{sweep}}(T,m):=m\cdot C_{\texttt{rej}}(2\delta_{m,T})=
    f(2\delta_{m,T})\frac{2mT}{m+1}\exp\{2\cta \delta_{m,T}\}.
\end{equation}
Equation \eqref{eq:gibbs_sweep_cost1} holds for all $m$ and $T$ and follows from \eqref{eq:comp_cost_psrs}; however, as the behaviour of $f(t)$ for small $t$ is not immediately transparent, to learn something about $C_{\texttt{sweep}}(T,m)$ when $\delta_{m,T}$ is small, we may use the fact that the acceptance probability of the rejection sampler approaches $1$ as the bridge duration decreases to $0$. This fact implies that for small enough $t$, $C_{\texttt{rej}}(t)\sim \ctd  t$ and thus if $\delta_{m,T}<\cte $ as $T\to\infty$ for some constant $\cte $ then
\begin{equation}\label{eq:gibbs_sweep_cost2}
  C_{\texttt{sweep}}(T,m) \sim \ctd  T,
\end{equation}
for some  $\ctd >0$. For instance, upon setting $m=\lfloor T \rfloor$, the cost in \eqref{eq:gibbs_sweep_cost1} becomes $\mathcal{O}(T)$ as $T\to\infty$. Contrast this with \eqref{eq:comp_cost_psrs} to see that the relative gain in efficiency, $C_{\texttt{rej}}(T)/C_{\texttt{sweep}}(T,m)$ grows exponentially in $T$ and suggests that blocking is to be preferred for large enough $T$. However, this ignores the costs associated with mixing; we address this in the next subsection.


\subsection{Cost of multiple sweeps}

Direct comparison of the exponential cost $C_{\texttt{rej}}(T)$ of direct rejection sampling (as given by Proposition \ref{prop:comp_cost}), with the linear cost $C_{\texttt{sweep}}(T,m)$ of a single sweep of a blocking scheme (as given by \eqref{eq:gibbs_sweep_cost2}), does not capture the remnant dependency structure introduced by the blocking scheme. In addition we need to consider the number of sweeps required to render this dependency negligible. In order to do that we first introduce the following notion

\begin{definition}\label{def:convergence_rate_L2}\citep{roberts1997updating}
  The \emph{[$\Ltwo$-]convergence rate $\rho$} of a Markov chain $\{X^{(n)};n=1,\dots,N\}$ with the transition kernel $P$ and an invariant density $\pi$ is defined as the minimum number for which for all square $\pi$-integrable functions $f$, and for all $r>\rho$
  \begin{equation*}
    \lVert P^nf-\pi(f) \rVert_{\Ltwo(\pi)}:=\int\left[P^nf(X^{(0)}) - \pi(f)\right]^2\pi(\dd X^{(0)}) \leq V_f r^n,
  \end{equation*}
  where $P^nf(X^{(0)}):=\mathbbm{E}_\pi[f(X^{(n)})|X^{(0)}]$, $\pi(f):=\mathbbm{E}_\pi[f(X)]$ and $V_f$ is a positive number that depends on $f$.
\end{definition}

We can now capture the cost of reducing the dependency on the past by considering the \emph{relaxation time}, denoted $\mathcal{T}=\mathcal{T}(T,m)$, and defined as:
\begin{equation}\label{eq:relaxation_time_vs_spectral_radius}
  \mathcal{T}= -\frac{1}{\log\left(\rho\right)}.
\end{equation}
It represents the time required by the underlying Markov chain to output a draw from its stationary distribution \citep{levin2017markov}. This makes it possible to compare $C_{\texttt{rej}}(T)$ with the expected computational cost of the blocked rejection sampler as follows:
\begin{equation}\label{eq:blocking_cost_template1}
  C_{\texttt{blocking}}(T,m):=
    \mathcal{T}(T,m)\cdot C_{\texttt{sweep}}(T,m).
\end{equation}
We will later consider the most appropriate choice of blocking scheme, and how to optimise $m$.

Instead of analysing the chain targeting the law $\condLaw$ it is sufficient to consider a related chain that targets the marginal law of the vector
\begin{equation}\label{eq:blocking_def_of_G}
\blockG:=(X_{k_1},\dots,X_{k_m})|(X_0,X_T)=\knots|(X_0,X_T),
\end{equation}
which we denote by $\GcondLaw$. To see this, notice that conditionally on the knots $\knots$ being distributed according to $\GcondLaw$, a path $X$ returned after a single Gibbs sweep of a blocking scheme is distributed exactly according to $\condLaw$. The object of interest becomes a Markov chain with a transition kernel $P$ denoting a single Gibbs sweep, and with stationary distribution $\pi=\GcondLaw$ \citep{roberts2001markov}.

Throughout, we additionally assume that the following condition holds, which makes the subsequent required calculations tractable.  
\begin{assumption}\label{as:gsn_blocks}
    The target law $\condLaw$ is such that $\blockG$ is a Gaussian process.
\end{assumption}
We discuss this key technical assumption in Section \ref{sec:numerical}, where we note that the established results seem to hold empirically more broadly.

Under Assumption \ref{as:gsn_blocks} and using either the lexicographic or checkerboard updating scheme, a single Gibbs step (i.e.\ an update $\blockG|_{\mathcal{B}_I}\sim\condLaw|_{\mathcal{B}_I\cap \{X_{k_1},\dots,X_{k_m}\}}(\cdot|\mathcal{K}_{-I})$) has a tractable, Gaussian transition density, and thus so does the entire Gibbs sweep $\blockG^{(n)}\mapsto\blockG^{(n+1)}$ with mean and covariance
\begin{equation*}
    \mu:=\mathbbm{E}[\blockG],\qquad\Sigma := \mathbbm{C}ov[\blockG].
\end{equation*}
As a consequence it is possible to explicitly characterise the transition kernel $P$, as follows.
\begin{lemma}\label{lem:n_step_transition_density}
    Under the lexicographic and checkerboard updating schemes, the $n$-step transition kernel $P^n$ of the Markov chain $\{\blockG^{(l)}\,;\,l=0,\dots\}$ is Gaussian, with mean and covariance matrix given respectively by:
    \begin{equation}\label{eq:blocking_mean_var}
        \mathbbm{E}[\blockG^{(l+n)}|\blockG^{(l)}]=B^{n}\blockG^{(l)}+(I-B)^{-1}(I-B^{n})b,\quad\mathbbm{C}ov[\blockG^{(l+n)}|\blockG^{(l)}]=\Sigma -B^n\Sigma(B^n)^\T,
    \end{equation}
    with $B\in\mathbbm{R}^{m\times m}$ and $b\in\mathbbm{R}^m$.
\end{lemma}

Under the lexicographic or the checkerboard updating schemes $\{\blockG^{(l)}\,;\,l=0,\dots\}$ is an AR$(1)$ process, and so the spectral radius $\rho_{\texttt{spec}}(B)$ of the matrix $B$ must satisfy $\rho_{\texttt{spec}}(B)<1$ for the process to converge, and equals the $\mathcal{L}^2$-convergence rate \citep{amit1991rates}. This connection extends to the random updating scheme. In the following lemma we derive the spectral radius of each blocking scheme as a function of $m$ and $T$, which aids in optimising their parameterisation and analysing their scaling.  We denote by $\Lambda:=\Sigma^{-1}$ the precision matrix of $\blockG$ and define 
\begin{equation*}
    A:=I-\mbox{diag}\{\Lambda_{11}^{-1},\dots,\Lambda_{mm}^{-1}\}\Lambda.
\end{equation*}

\begin{lemma}\label{lem:general_L2_conv}\citep{roberts1997updating}
    Under the checkerboard and lexicographic updating schemes, the spectral radius of the matrix $B$ and the $\Ltwo$-convergence rate of a blocked rejection sampler coincide. More explicitly, under the checkerboard, lexicographic, and random updating schemes respectively the $\Ltwo$-convergence rates ($\rho_{\texttt{check}}$, $\rho_{\texttt{lex}}$, and $\rho_{\texttt{rand}}$ resp.) are equal to:
    \begin{align*}
         \rho_{m,T}:=\rho_{\texttt{check}}=\rho_{\texttt{lex}}=\rho_{\texttt{spec}}(B_{\texttt{lex}})=\rho_{\texttt{spec}}(B_{\texttt{check}})&=\lambda_{\texttt{max}}^2(A), &  \rho_{\texttt{rand}}&=\left[\frac{m-1+\lambda_{\texttt{max}}(A)}{m}\right]^m,
    \end{align*}
    where $\lambda_{\texttt{max}}(A)$ denotes the maximum eigenvalue of the matrix $A$ and where we write $B_{\texttt{check}}$ (resp. $B_{\texttt{lex}}$) to denote a matrix $B$ corresponding to the checkerboard (resp. lexicographic) updating scheme.
\end{lemma}
$\lambda_{\texttt{max}}(A)$ can be found more explicitly by exploiting the close connection between the precision matrix $\Lambda$ and the matrix of partial correlations (given precisely in \eqref{eq:corr_to_prec_connect}, in Appendix \ref{sec:proofs}).

\begin{theorem}\label{thm:general_explicit_L2_conv}
    We have
    \begin{equation*}
        \lambda_{\texttt{max}}(A)=2\lvert c(\delta_{m,T})\rvert\cos\left(\frac{\pi}{m+1}\right),
    \end{equation*}
    with $c(\delta_{m,T}):=\mathbbm{C}orr(X_\delta,X_{2\delta}|X_0,X_{3\delta})$. In particular:
    \begin{equation}\label{eq:rho_mT}
        \rho_{m,T}
            =
                4c^2(\delta_{m,T})\cos^2\left(\frac{\pi}{m+1}\right),\qquad
        \rho_{\texttt{rand}}
            =
                \left[
                    \frac{
                        m
                        - 1
                        + 2\lvert
                            c(\delta_{m,T})
                        \rvert
                        \cos\left(
                            \frac{\pi}{m+1}
                        \right)
                    }
                    {m}
                \right]^m.
    \end{equation}
\end{theorem}

The form of $c(\delta_{m,T})$ will, in general, depend on the type of a Gaussian process that is being considered. In the following corollaries we present more explicit versions of the statements from Theorem \ref{thm:general_explicit_L2_conv} for the two choices of $\law$: scaled Brownian motion $\sigma W$, with $\sigma>0$; and, the Ornstein--Uhlenbeck process. Without loss of generality we centre the latter at $0$:
\begin{equation}\label{eq:blocking_sde_ou}
    \dd X_t = -\theta X_t\dd t + \sigma\dd W_t,\quad X_0=x_0,\quad t\in[0,T].
\end{equation}

\begin{corollary}\label{cor:bm_L2_conv}
  If $\law$ is the law of a scaled Brownian motion $\sigma W$, $\sigma > 0$, then:
    \begin{equation*}
        \rho_{m,T}
            =
                \cos^2\left(\frac{\pi}{m+1}\right),
            \quad
                \rho_{\texttt{rand}}
            =
                \left[
                    \frac{
                    m
                    - 1
                    + \cos\left(
                        \frac{\pi}{m+1}
                    \right)
                    }
                    {m}
                \right]^m.
    \end{equation*}
    In particular, independently of $T$, as $m\to\infty$
    \begin{equation*}
        \rho_{m,T}
            =
                1-\left(\frac{\pi}{m+1}\right)^2+\mathcal{O}(m^{-4}),
        \qquad\rho_{\texttt{rand}}
            =
                1-\frac{1}{2}\left(\frac{\pi}{m+1}\right)^2+\mathcal{O}(m^{-4}).
    \end{equation*}
\end{corollary}

\begin{corollary}\label{cor:ou_L2_conv}
    If $\law$ is the law of the Ornstein--Uhlenbeck process \eqref{eq:blocking_sde_ou}, then:
    \begin{equation*}
        \rho_{m,T}
            =
                \cos^2\left(\frac{\pi}{m+1}\right)\sech^2(\theta \delta_{m,T}),
            \quad
                \rho_{\texttt{rand}}
            =
                \left[
                    \frac{
                    m
                    - 1
                    + \cos\left(\frac{\pi}{m+1}\right)\sech(\theta \delta_{m,T})
                    }
                    {m}
                \right]^m.
    \end{equation*}
    In particular, when $\delta_{m,T}=\delta$ is set to a constant, as $m,T\to\infty$:
    \begin{align*}
    \rho_{m,T} &= \sech^2(\theta\delta)\left[1-\left(\frac{\pi}{m+1}\right)^2+\mathcal{O}(m^{-4})\right], &
         \rho_{\texttt{rand}}
        & = e^{\sech(\theta \delta)-1}\left[1-\frac{(1-\sech(\theta \delta))^2}{2(m+1)}+\mathcal{O}(m^{-2})\right].
    \end{align*}
\end{corollary}
\begin{remark}
Results obtained by \citet{pit:she:1999:JTSA}, who studied the discrete-time first-order autoregressive process $\alpha_t = \phi \alpha_{t-1} + \eta_t$, $\eta_t \sim \mathcal{N}(0,\sigma^2)$, observed with Gaussian noise, are closely related to Corollaries \ref{cor:bm_L2_conv} and \ref{cor:ou_L2_conv}. In the context of their model, where $c(\delta_{m,T}) = \phi/(1+\phi^2)$, they derive the expression \eqref{eq:rho_mT} for $\rho_{m,T}$ as well as bounds on $\rho_{m,T}$ which exhibit the same asymptotic behaviour as in Corollary \ref{cor:ou_L2_conv}.
\end{remark}

We can now combine the above results with \eqref{eq:relaxation_time_vs_spectral_radius} to find the relaxation time:

\begin{theorem}\label{thm:number_of_needed_steps}
   Suppose we use one of the checkerboard, lexicographic, and random updating schemes. If $\law$ is the law of a scaled Brownian motion $\sigma W$, then we have:
    \begin{equation*}
        \mathcal{T}(m)=\mathcal{O}(m^2), \qquad m\to\infty.
    \end{equation*}
    If $\law$ is the law of the Ornstein--Uhlenbeck process in \eqref{eq:blocking_sde_ou}, and additionally the sequence $T(m)$ is chosen so that $m=\ctc T$ for some constant $\ctc>0$, then we have:
    \begin{equation*}
        \mathcal{T}(m)=\mathcal{O}(1), \qquad m\to\infty.
    \end{equation*}
\end{theorem}

\begin{remark}
 Note that if $\law$ is the law of the Ornstein--Uhlenbeck process in \eqref{eq:blocking_sde_ou} then Theorem \ref{thm:number_of_needed_steps} holds for only the sequence $T(m)$ where $m=\ctc T$, but in the case of scaled Brownian motion there is no such constraint; see Remark \ref{rem:oubetter}.
\end{remark}

\begin{remark}
 From the proof of Theorem \ref{thm:number_of_needed_steps}, one can show that for the Ornstein--Uhlenbeck process \eqref{eq:blocking_sde_ou}:
  \begin{equation*}
  	\mathcal{T}(m) = -\frac{1}{\log(\rho_{m,T})}\to -\frac{1}{2\log(\sech(\theta\delta))}, \qquad m\to\infty.
  \end{equation*}
  This provides insight into the influence of $\theta$ and $\delta$ on mixing.
\end{remark}

We can minimize the cost of blocking $C_{\texttt{blocking}}(T,m)$ over the remaining parameter, $m$, using Theorem \ref{thm:number_of_needed_steps} and \eqref{eq:blocking_cost_template1}. This leads to Theorem \ref{thm:main_thm_comp_cost}, which is the main result of this paper (as presented in Section \ref{sec:intro}, with accompanying proof in Appendix \ref{sec:proofs}).


\section{Numerical experiments}\label{sec:numerical}

Consider a target process defined to be the solution of the following stochastic differential equation (with law $\law$): 
\begin{equation}\label{eq:sine_diff}
  \dd X_t = (2-2\sin(8X_t))\dd t + \frac{1}{2}\dd W_t,\quad X_0=0,\quad t\in[0,T].
\end{equation}

This diffusion exhibits highly multimodal behaviour, and so in practice it is challenging to simulate trajectories of $\law$ (and in particular the conditioned bridge law $\condLaw$ over large time horizons). It is possible to simulate trajectories exactly by means of path-space rejection sampling (as detailed in Appendix \ref{sec:psrs}). However, $X$ \emph{is not} a Gaussian process (it violates Assumption \ref{as:gsn_blocks}), and so Theorem \ref{thm:main_thm_comp_cost} does not hold in a rigorous sense. As such \eqref{eq:sine_diff} makes an interesting case to investigate the practical limitations of Theorem \ref{thm:main_thm_comp_cost}. Because it is not an ergodic diffusion, out of the two theoretical results from Theorem \ref{thm:main_thm_comp_cost} the ones for the Brownian motion are expected to be more relevant.  As we show below, the empirical results would suggest the theory holds more broadly.

We consider six problems (increasing in difficulty) of simulating paths according to the laws $\law^{(T,x_0,x_T)}$, with parameters $\law^{(0.2,0,0.1)}$, $\law^{(0.4,0,0.85)}$, $\law^{(0.5,0,0.85)}$, $\law^{(1,0,0.95)}$, $\law^{(2,0,2.5)}$, $\law^{(4,0,4.85)}$. The values of the end-points were chosen by fixing $T$, simulating multiple paths according to \eqref{eq:sine_diff} and picking $x_T$ to be some point in the vicinity of the (largest) mode as these are the bridges we will most commonly be interested in. For $T=0.2$, the plotted paths resemble Brownian bridges, but as $T$ increases the non-linear dynamics become pronounced: the diffusion is effectively attracted to a ladder of values and it is repelled at the intermediate points, leading to multimodal behaviour of the trajectories. Drawing paths from the last three laws using path-space rejection sampling but without blocking (an \emph{unmodified rejection sampler}) is computationally infeasible.

For each of the six examples we ran a blocked rejection sampler with checkerboard updating scheme for $10^5$ iterations and with various numbers of knots. For the first three problems we also employed an unmodified rejection sampler. We recorded the time required to sample a single path (which for a blocked rejection sampler is counted as one execution of the inner \texttt{for-loop} of Algorithm \ref{alg:blocking}) and plotted it in Figure \ref{fig:blocking_sin_elapsed} against the number of used knots. Code sufficient for reproducing these results can be found at \url{https://github.com/mmider/blocking}. 

\begin{figure}[ht]
\begin{center}
\includegraphics[width=0.8\textwidth]{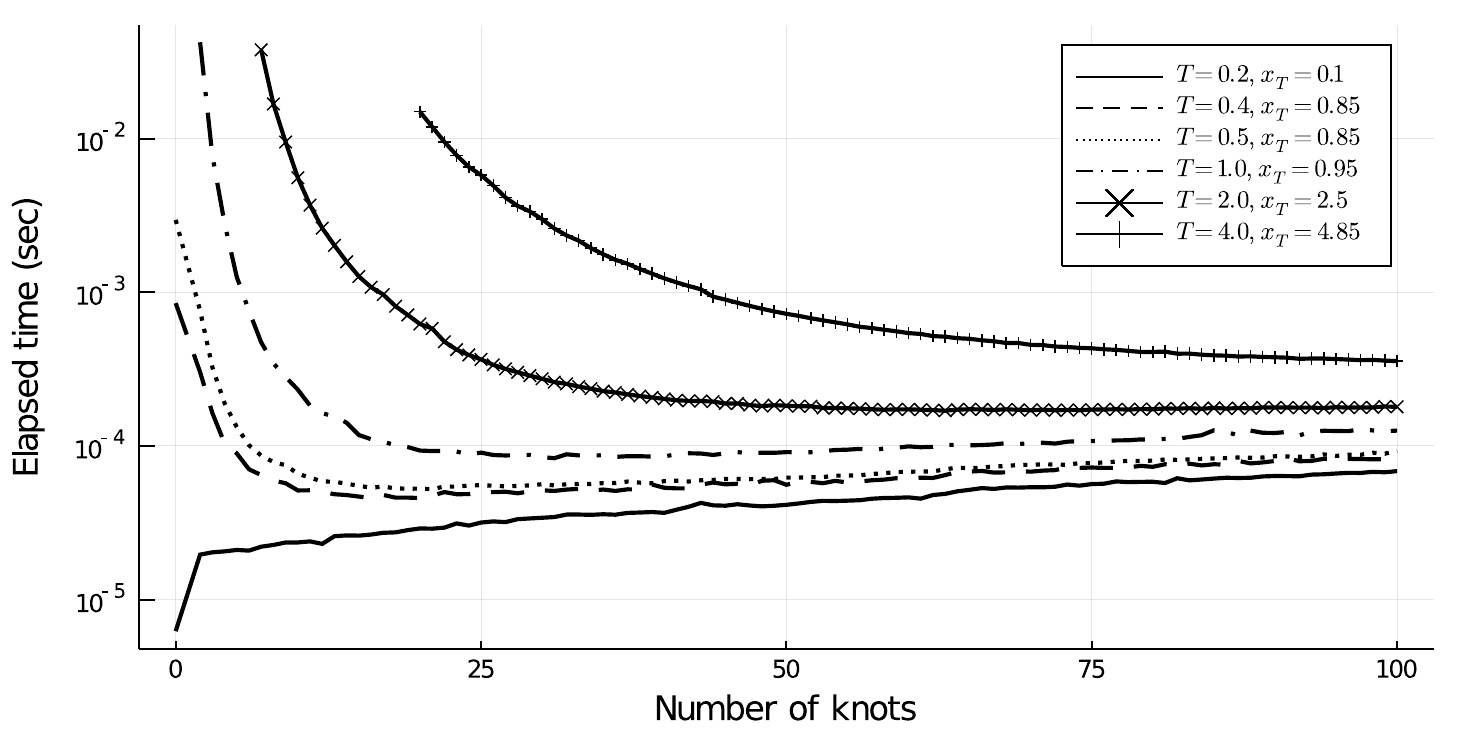}
\caption{Time (in seconds; log-transformed) required to sample a single path of the sine diffusion \eqref{eq:sine_diff} as a function of the number of used knots.} \label{fig:blocking_sin_elapsed}
\end{center}
\end{figure}

For $T=0.2$ the unmodified rejection sampler clearly outperforms any blocking scheme. This is unsurprising as paths under $\law^{(0.2,0,0.1)}$ closely resemble Brownian bridges (and indeed every diffusion behaves as a drifted Brownian motion on a small-enough time-scale). However, as $T$ increases, this pattern changes and blocking reduces the cost of obtaining any single sample path. In particular, notice a steep, exponential reduction in cost that is especially pronounced for $(T,x_T)=(1,0.95)$ (this would be illustrated even more emphatically by $(T,x_T)=(2,2.5)$ and $(T,x_T)=(4,4.85)$ had the corresponding experiments with a lower number of knots been run; however, their costs are prohibitively high and had to be omitted).

Figure \ref{fig:blocking_sin_elapsed}, though helpful in confirming Proposition \ref{prop:comp_cost}, does not take into account the cost due to decreased speed of mixing---the main motivation for the developments presented in Section \ref{sec:computational_cost_of_blocking}. To incorporate also this cost we plot in Figure \ref{fig:blocking_sin_ta_ess} the time-adjusted effective sample size (taESS), with
\[
\text{taESS}:=[\text{effective sample size}]/[\text{elapsed time in seconds to sample an entire chain}]
\]
(and ESS was computed according to \citet[Section 11.5]{gelman2013bayesian}) against the \mbox{(half-)} length of blocks (i.e.~$\delta_{m,T}$). As defined in Figure \ref{fig:blocking_sin_ta_ess}, taESS is approximately equal to a number of independent samples that can be drawn in one second. Clearly, the larger taESS is the more efficient the algorithm is.

\begin{figure}[ht]
\begin{center}
\includegraphics[width=0.8\textwidth]{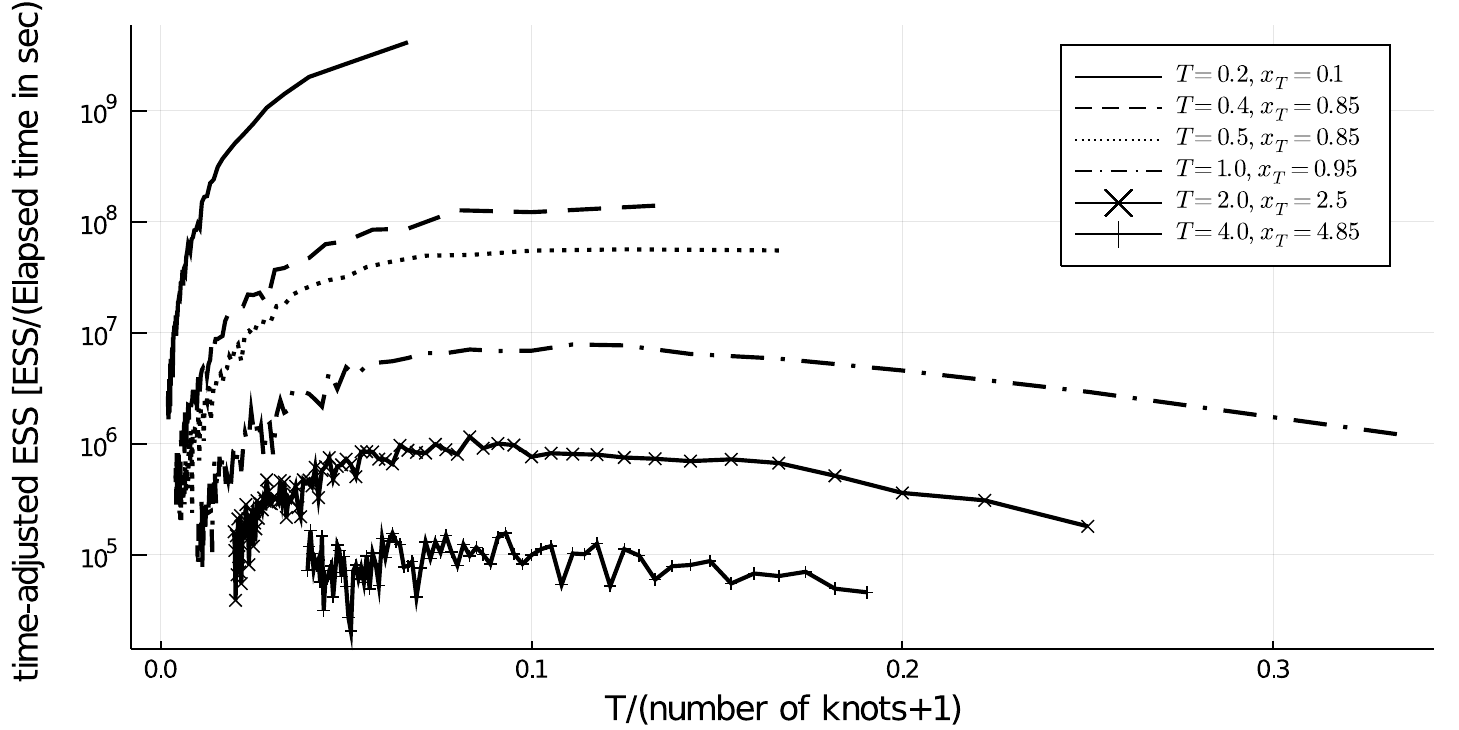}
\caption{Time-adjusted effective sample size 
vs half-length of blocks (i.e.~$\delta_{m,T}$).} \label{fig:blocking_sin_ta_ess}
\end{center}
\end{figure}

First, for any experiment we expect there to be a point for which increasing the number of knots any further will only lead to a decrease in taESS---this corresponds to all costs being dominated by the cost due to a slowdown in mixing and it is clearly illustrated by sharp dips of curves on the left side of Figure \ref{fig:blocking_sin_ta_ess}. Second, for examples for which the target law is sufficiently different from the law of Brownian bridges we expect that some level of blocking will improve the overall computational cost. This is also confirmed by the declines of taESS curves toward the right side of Figure \ref{fig:blocking_sin_ta_ess}. We note that under the most difficult sampling regimes it was impractical to run the algorithm with even fewer blocks due to excessive execution times---had the examples been run and the curves continued, the decline in performance would have been even starker. Additionally, Figure \ref{fig:blocking_sin_ta_ess} is suggestive of there being an optimal value of $\delta_{m,T}$ (somewhere around $\delta_{m,T}\approx 0.1$), that is almost independent of $T$ and $m$ and that yields the highest taESS in each experiment. This is consistent with the results of Section \ref{sec:computational_cost_of_blocking}, where an optimal number of knots was found to be $m= \ctc T$ for some $\ctc >0$, which implies the claim about the dependence of the optimal $\delta_{m,T}$ on $T$ and $m$.

Finally, we verify the bound from \eqref{eq:cost_of_blocking_explicit} empirically. To this end, notice that taESS$^{-1}$ is approximately equal to the amount of time needed to obtain a single independent sample. This is consistent with the characterisation of the computational cost of a blocked rejection sampler as given in \eqref{eq:blocking_cost_template1}. 
Theorem \ref{thm:main_thm_comp_cost} asserts that this cost scales at a cubic rate in the duration of the bridge, so long as $\delta_{m,T}$ is set to a constant when $T\to\infty$. Consequently, taESS$(T)$ should be at most a cubic function of $T$ and if plotted on a log-log scale, this would be equivalent to taESS$(T)$ tracking some line with slope $3$. Figure \ref{fig:blocking_sin_cost_pred} gives this precise plot, showing that the prediction \eqref{eq:cost_of_blocking_explicit} is indeed satisfied.

\begin{figure}[ht]
\begin{center}
\includegraphics[width=0.8\textwidth]{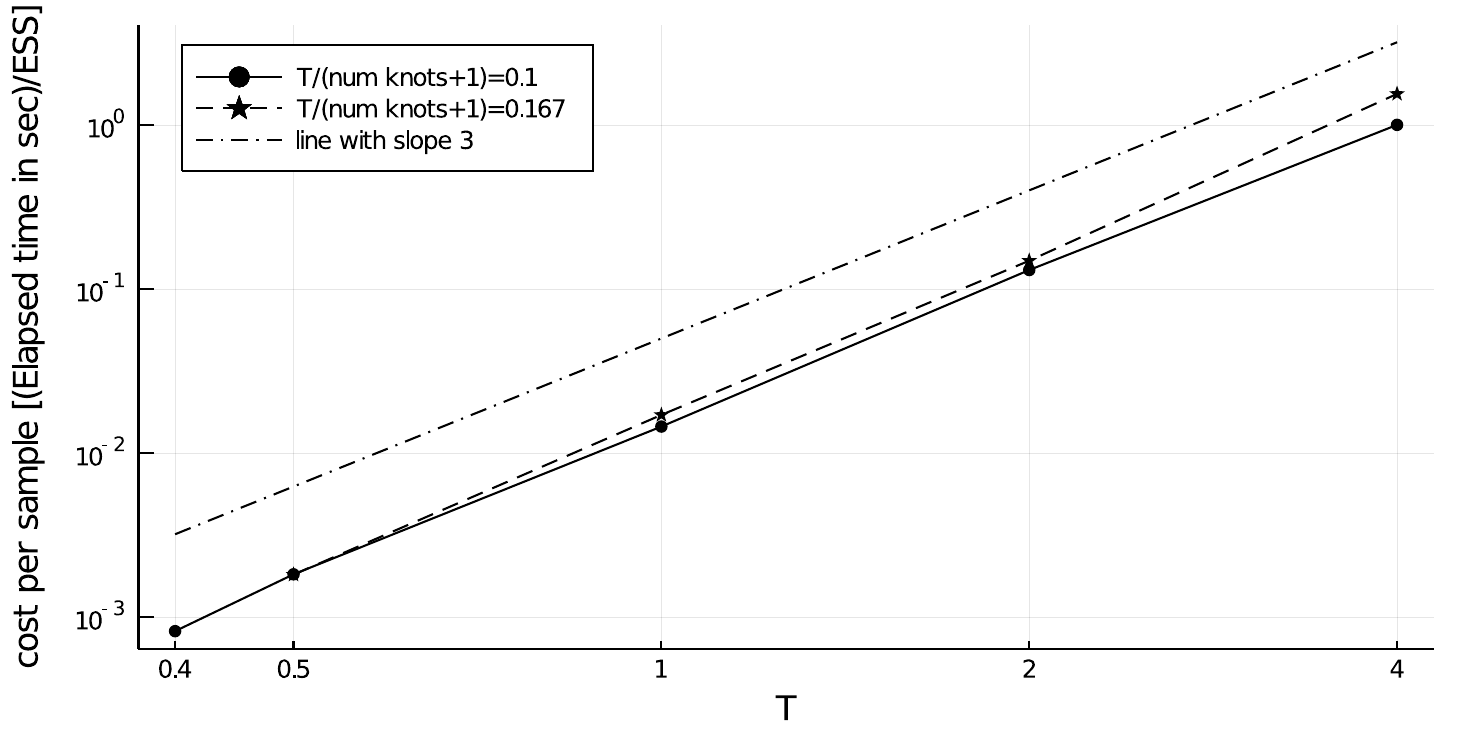}
\caption{Computational cost as a function of time for the sine example.} \label{fig:blocking_sin_cost_pred}
\end{center}
\end{figure}


\section{Discussion}\label{sec:discussion}

In this article we have analysed and provided practical guidance for using blocking schemes when conducting Bayesian inference for discretely observed diffusions. We achieved this by studying the computational cost of diffusion bridge sampling algorithms. We have shown rigorously that the computational cost of rejection sampling on path-space (modified with blocking) targeting the law of scaled Brownian motion scales as $\mathcal{O}(T^3)$ as $T\to\infty$, and as $\mathcal{O}(T)$ in the case of the Ornstein--Uhlenbeck process, so long as the number of equidistant anchors is $m=\ctc T$ (for some $\ctc >0$). In Remark \ref{rem:oubetter} we discussed the practicality of exploiting the computational saving achievable in the case of the Ornstein--Uhlenbeck process. Furthermore, using the example of a non-linear sine diffusion we provide empirical evidence which would suggest that the conclusions about Brownian motion hold also for non-ergodic diffusions outside of a restrictive class of Gaussian processes.

Our theory indicates that choosing too few knots results in the computational cost being dominated by the exponential cost for imputing diffusion bridges between successive knots (see Proposition \ref{prop:comp_cost}). As such our guideline of choosing $m=\ctc T$, (for some $\ctc >0$) is useful for ensuring the robustness of blocking schemes and a reasonable heuristic for practitioners. Note that although choosing too many knots is likely to be penalized less than choosing too few, choosing an excessive number of knots can negatively impact the mixing of the underlying chain.  

Naturally, for more general target laws $\law$ it might be useful to consider using irregularly spaced anchors (and so relaxing Assumption \ref{as:equidistant_grid}). Heuristically, we may wish to place more knots in areas in which the proposal law does not approximate the target law well. Developing more general theory to support the use of an irregular spacing of anchors is likely to require more knowledge of the specific diffusion under study. Of course, from a methodological perspective this motivates future research looking at how to place knots by assessing proposal-target discrepancy, or developing adaptive schemes.

Finally, it is worth recalling that within the context of Bayesian inference for discretely observed diffusion processes, the full chain in this setting is a Gibbs sampler that alternates between updating the unknown parameters and imputing the unobserved path. Since the mixing time of the unobserved path influences the mixing time of the parameter chain, then in light of the work in this paper it may as a future extension also be possible to study the mixing behaviour of the parameter chain.


\appendix
\setcounter{equation}{0}
\renewcommand{\theequation}{\thesection.\arabic{equation}}


\section{Rejection sampling on path-space}\label{sec:psrs}

In this article we have restricted our attention to the class of diffusion bridges which can be sampled by means of \emph{path-space rejection sampling}. In particular, to sample from $\condLaw$ we sample trajectories from an accessible and absolutely continuous \emph{proposal law} (denoted $\mathbbm{Q^*}$), and accept with probability proportional to the Radon-Nikod\'{y}m derivative of $\condLaw$ to $\mathbbm{Q^*}$ \citep{beskos2005exact, beskos2006retrospective, beskos2008factorisation, b:pjr16}. To find an appropriate $\mathbbm{Q^*}$ we impose the following common assumption \citep[Section 4.4]{kloeden1992numerical}
\begin{assumption}\label{as:lamperti}
    There exists $\eta:\RR^d\to\RR^{d}$ such that $\nabla\eta = \sigma^{-1}$.
\end{assumption}
Under Assumption \ref{as:lamperti} the  process $Y:=\{\eta(X_t),t\in[0,T]\}$ satisfies the following stochastic differential equation,
\begin{equation*}
    \dd Y_t = \alpha(Y_t)\dd t + \dd W_t,\quad Y_0=y_0:=\eta(x_0),\quad t\in[0,T],
\end{equation*}
for a known, closed-form drift $\alpha$. With unit volatility, the law of $Y$ is now absolutely continuous with respect to Brownian motion, and so Brownian motion is a viable proposal law for sampling from $\mathbbm{P}$ (and the law induced by the Brownian bridge is a viable proposal for $\condLaw$).

In order to avoid unnecessary inflation of notation we assume throughout the article that $\sigma\equiv 1$ in \eqref{eq:master_sde}, so that $\alpha\equiv b$, $\eta$ becomes an identity map, and $X\equiv Y$. The general case of $\sigma$ (that satisfies Assumption \ref{as:lamperti}) follows without additional effort. 

\begin{assumption} \label{as:alpha}
    $\alpha$ is at least once continuously differentiable.
\end{assumption}
\begin{assumption}\label{as:potential}
    There exists a potential function $A:\RR^d\to\RR$ such that $\nabla A = \alpha$.
\end{assumption}
\begin{assumption}\label{as:lower_bound}
   The function $\phi(y):=\frac{1}{2}\left(\lVert\alpha(y)\rVert^2 + \Delta A(y)\right)$ is bounded from below by some $\Phi:=\inf\{\phi(y):y\in\RR^d\}\in\RR$.
\end{assumption}
Under assumptions \ref{as:alpha}--\ref{as:lower_bound} we have \citep[Section 3]{beskos2005exact}:
\begin{equation}
\label{eq:pdef}
    \frac{\dd\condLaw}{\dd\mathbbm{Q^*}}\left(\{\eta^{-1}(Y_t),t\in[0,T]\}\right)\propto\exp\left\{ - \int_0^T(\phi(Y_t)- \Phi)\dd t \right\}=:\texttt{p}(Y)\leq 1.
\end{equation}

It follows that sampling from $\condLaw$ can be accomplished using Algorithm \ref{alg:psrs_simple}. Note that computing the integral in \eqref{eq:pdef} required for Algorithm \ref{alg:psrs_simple} can be achieved either (i) approximately, by simulating a candidate $Y^{\circ}$ over a fine mesh and computing the integral in \eqref{eq:pdef} numerically; or (ii) exactly, via an additional randomisation step that utilises a Poisson point process \citep{beskos2005exact, beskos2006retrospective, beskos2008factorisation}. We refer to these two methods as, respectively, \emph{approximate} and \emph{exact} path-space rejection samplers. If we contrast $U\sim\texttt{Bern}(\texttt{p}(Y))$ with that of $\widetilde{U}\sim\texttt{Bern}(\tilde{\texttt{p}}(Y))$ where $\tilde{\texttt{p}}$ is the (unbiased) estimator resulting from the additional randomisation step, then since $\mathbbm{E}[\tilde{\texttt{p}}]=\texttt{p}$ we have $U\overset{d}{=}\widetilde{U}$ and so $\mathbbm{V}(U)=\mathbbm{V}(\widetilde{U})$, and there is no additional introduction of variance from the use $\tilde{\texttt{p}}$ in place of $\texttt{p}$ \citep{lat:etal:2011}.

\begin{algorithm}[H]
    \caption{Rejection sampling on path-space}\label{alg:psrs_simple}
        \While{\texttt{True}}{
            Draw $Y^{\circ}\sim\mathbbm{W}^{(T,\eta(x_0),\eta(x_T))}$, i.e. a $d$-dimensional Brownian bridge joining $\eta(x_0)$ and $\eta(x_T)$ on $[0,T]$\;
            Draw $U\sim\texttt{Unif([0,1])}$\;
            \If{$U \leq \texttt{p}(Y^{\circ})$}{
                Set $X\leftarrow \{\eta^{-1}(Y^{\circ}_t),t\in[0,T]\}$\;
                \Return{$X$}
            }
        }
\end{algorithm}

To prove Proposition \ref{prop:comp_cost}, which quantifies the computational cost of Algorithm \ref{alg:psrs_simple}, we impose the following natural assumptions on the cost of simulating a proposal trajectory. 

\begin{assumption}\label{as:independence_of_cost}
    The cost of generating any proposal sample $X$ is independent of the value of $\texttt{p}(X)$ as defined in \eqref{eq:pdef}.
\end{assumption}

\begin{assumption}\label{as:x_cost}
    The cost $\mathbbm{c}(X)$ of simulating $X$ has expectation growing linearly in T: $\mathbbm{E}\left[\mathbbm{c}(X)\right]=\ctd  T$, $\ctd \in\RR_+$.
\end{assumption}

Assumptions \ref{as:independence_of_cost} and \ref{as:x_cost} are always satisfied if rejection sampling is performed with the approximate method described above, so long as the mesh width is kept constant as $T\to\infty$. For the exact method, Assumption \ref{as:independence_of_cost} will in general be violated (for instance, if the number of simulated Poisson points is $0$, then conditional on this information $\texttt{p}(X)=1$ a.s.) but for $T\to\infty$ it is a reasonable approximation. Assumption \ref{as:x_cost} is satisfied if $\phi$ is bounded.

We can now derive the cost of a single draw using a path-space rejection sampler as follows:
\begin{lemma}\label{lem:for_comp_cost}
  Under Assumptions \ref{as:lamperti}--\ref{as:x_cost}:
  \begin{equation*}
    C_{\texttt{rej}}(T) = \ctg \frac{q_T(x_0,x_T)}{p_T(x_0,x_T)} Te^{-\Phi T},
  \end{equation*}
  where $\ctg >0$ is some constant independent of $T$, $p_T(x_0,x_T)$ is the transition density under $\law$ for going from $x_0$ to $x_T$ over the interval $[0,T]$ and $q_T(x_0,x_T)$ is the same transition density, but under the proposal law $\mathbbm{Q}$ instead.
 \end{lemma}
\begin{proof}
  Denote by $X^{(i)}$, $i\in\{1,2,\dots\}$, independent samples from $\mathbbm{Q^*}$ and by $\mathbbm{c}(X^{(i)})$ the cost of sampling path $X^{(i)}$, $i\in\{1,2,\dots\}$. Rejection sampling requires a geometrically distributed number of simulations (with a randomly distributed parameter at each trial), so its expected cost is
    \begin{align}
            \mathcal{C}_{\texttt{rej}}(T) &:= \mathbbm{E}\left[\sum_{i=1}^\infty \left(\sum_{j=1}^i\mathbbm{c}(X^{(j)})\right) \cdot \texttt{p}(X^{(i)})\prod_{j=1}^{i-1}\left(1-\texttt{p}(X^{(j)})\right)\right] \notag\\
            &= \mathbbm{E}\left[\mathbbm{c}(X^{(1)})\right]\sum_{i=1}^\infty i \mathbbm{E}\left[ \texttt{p}(X^{(i)}) \right]\prod_{j=1}^{i-1}\left(1-\mathbbm{E}\left[\texttt{p}(X^{(j)})\right]\right)\notag\\
            &=\frac{\mathbbm{E}\left[\mathbbm{c}(X)\right]}{\mathbbm{E}\left[\texttt{p}(X)\right]},\label{eq:Crej}
    \end{align}
    where the measures with respect to which the expectations above are taken should be clear from the context (and include Brownian bridge measures, products of Brownian bridge measures, and any additional randomness needed to simulate events of probability $\texttt{p}(X^{(i)})$).
    We now have:
        \begin{align}
            \mathbbm{E}_{\mathbbm{Q^*}}\left[\texttt{p}(X)\right] &= \mathbbm{E}_{\mathbbm{Q^*}}\left[ \exp\left\{ -\int_0^T(\phi(X_t)-\Phi)\dd t \right\} \right] \label{eq:exp_ar_1}\\
            &= \mathbbm{E}_{\mathbbm{Q^*}}\left[ \exp\left\{ \left[A(X_T) - A(X_0)\right] - \left[A(X_T) - A(X_0)\right] - \int_0^T (\phi(X_t)-\Phi)\dd t \right\} \right]\notag\\
            &= \exp\left\{ -A(X_T) + A(X_0) +\Phi T \right\}\frac{p_T(x_0,x_T)}{q_T(x_0,x_T)} \mathbbm{E}_{\mathbbm{Q^*}}\left[ \frac{\dd\law^\star}{\dd\mathbbm{Q^*}}(X)  \right] \notag\\
            &= \cth \frac{p_T(x_0,x_T)}{q_T(x_0,x_T)}e^{\Phi T}\label{eq:exp_ar_last},
        \end{align}
    where $\cth :=\exp\left\{ -A(x_T) + A(x_0)\right\}$ and where the third equality followed from \citet[Eq (3.1)]{dacunha1986estimation}. The result now follows by substituting \eqref{eq:exp_ar_last} into \eqref{eq:Crej} and noting that, by assumption \ref{as:x_cost}, we have $\mathbbm{E}\left[\mathbbm{c}(X)\right]=\ctd  T$. 
\end{proof}

To better understand the scaling with $T$ of the ratio of transition densities under the laws $\law$ and $\mathbbm{Q}$ in Lemma \ref{lem:for_comp_cost} we impose the following final assumption, which 
allows us to establish Lemmata \ref{lem:for_comp_cost2} and \ref{lem:l_negative} required for proving Proposition \ref{prop:comp_cost}.

\begin{assumption}\label{as:ergodicity}
  The target diffusion is ergodic and defined on $\RR^d$.
\end{assumption}

\begin{lemma}\label{lem:for_comp_cost2}
  Under Assumption \ref{as:ergodicity}, for
  \begin{equation*}
    f:T\to\frac{q_T(x_0,v)}{p_T(x_0,v)},\quad v\in\RR^d,
  \end{equation*}
  we have that $f(T)\sim T^{-d/2}$ as $T\to\infty$ and $d$ denotes the dimension of the process.
\end{lemma}
\begin{proof}
  $q$ and $p$ are well-behaved densities which are bounded and bounded away from zero (see for example \cite{rogers1985smooth}). This implies that $f$ is continuous. As the target diffusion is ergodic, $p_T(x_0,v)\to \hat{p}(v)$ as $T\to\infty$, where $\hat{p}$ is the stationary density of the diffusion law. On the other hand $q_T(x_0,v)$ is just a Gaussian density with variance $T^{d}I$, which for $T\to\infty$ behaves as $\sim T^{-d/2}$. 
\end{proof}

\begin{lemma}\label{lem:l_negative}
  Assumption \ref{as:ergodicity} implies that $\Phi<0$.
 \end{lemma}
 \begin{proof}
From Lemma \ref{lem:for_comp_cost2} we have that the RHS of \eqref{eq:exp_ar_last} is $\sim T^{d/2}e^{\Phi T}$ for $T\to\infty$. As the LHS of \eqref{eq:exp_ar_last} represents an expected probability, we must have $\Phi < 0$ for this expression to take values in $[0,1]$. 
\end{proof}

We are now in a position to prove Proposition \ref{prop:comp_cost}.

\begin{proof}[Proof of Proposition \ref{prop:comp_cost}]
  This follows directly by combining Lemmata \ref{lem:for_comp_cost}, \ref{lem:for_comp_cost2} and \ref{lem:l_negative}. 
\end{proof}


\section{Proofs}\label{sec:proofs}
\setcounter{equation}{0}

\begin{proof}[Proof of Lemma~\ref{lem:n_step_transition_density}]
    The chain $\{\blockG^{(l)}\,;\,l=0,\dots\}$ coincides with the chains considered in \cite{roberts1997updating}. In particular, the $1$-step transition kernel under lexicographic and checkerboard updating schemes is stated explicitly as \citet[Lemma 1]{roberts1997updating}. We provide a proof for completeness.
    
    $\{\blockG^{(l)}\,;\,l=0,\dots\}$ behaves like an AR$(1)$ process, therefore
    \begin{equation*}
        \blockG^{(l+1)}=B\blockG^{(l)}+\epsilon,\qquad \epsilon\sim\mathcal{N}(b,V),
    \end{equation*}
    for some $B$, $b$, and $V$ and
    \begin{equation}\label{eq:AR_process}
        \blockG^{(l+n)}=B^n\blockG^{(l)}+\epsilon^{(n)},\qquad \epsilon^{(n)}\sim\mathcal{N}(b^{(n)},V^{(n)}),
    \end{equation}
    with $b^{(n)}:=(I+B+\dots+B^{n-1})b=(I-B)^{-1}(I-B^{n})b$ and some $V^{(n)}$ that we are about to derive. Under either scheme $b$ and $B$ can be found in closed form (which we omit for brevity). If the chain has reached stationarity, i.e. if $\blockG^{(l)}\sim\mathcal{N}(\mu,\Sigma)$, then also $\blockG^{(l+n)}\sim\mathcal{N}(\mu,\Sigma)$. On the other hand, if $\blockG^{(l)}\sim\mathcal{N}(\mu,\Sigma)$, then by \eqref{eq:AR_process}
    \begin{equation*}
        \blockG^{(l+n)}|\blockG^{(l)}
            \sim
                \mathcal{N}\left(
                    B^{n}\blockG^{(l)}
                    + (I-B)^{-1}(I-B^{n})b ,
                    B\Sigma B^\T
                    + V^{(n)}
                \right).
    \end{equation*}
    Consequently:
    \begin{equation*}
        \Sigma = B^n\Sigma (B^n)^\T + V^{(n)},
    \end{equation*}
    and this yields \eqref{eq:blocking_mean_var}. 
\end{proof}

\begin{proof}[Proof of Lemma \ref{lem:general_L2_conv}]
    Since the chain $\{\blockG^{(l)}\,;\,l=0,\dots\}$ coincides with the chains considered in \cite{roberts1997updating}, the statement of \citet[Theorem 1]{roberts1997updating} applies under checkerboard and lexicographic updating schemes: i.e. the $\Ltwo$ convergence rates under the two regimes are given by $\rho_{\texttt{check}}=\rho_{\texttt{spec}}(B_{\texttt{lex}})$ and $\rho_{\texttt{lex}}=\rho_{\texttt{spec}}(B_{\texttt{lex}})$ respectively. Due to tridiagonal structure of the precision matrix $\Lambda$ (which follows from the Markov property of the process $\blockG$; see also a short explanation in the proof of Theorem \ref{thm:general_explicit_L2_conv} that leads up to \eqref{eq:matrix_A}), \citet[Corollary 3]{roberts1997updating} implies that the two spectral radii coincide, i.e. $\rho_{\texttt{spec}}(B_{\texttt{lex}})=\rho_{\texttt{spec}}(B_{\texttt{check}})$. By the same token, \citet[Theorem 5]{roberts1997updating} applies as well, yielding $\rho_{\texttt{spec}}(B_{\texttt{check}})=\lambda_{\texttt{max}}^2(A)$. Finally, the $\Ltwo$ convergence rate of the random updating scheme follows from \citet[Theorem 2]{roberts1997updating}. 
\end{proof}

\begin{proof}[Proof of Theorem \ref{thm:general_explicit_L2_conv}]
The precision matrix $\Lambda$ of any random vector $\blockG$ with non-degenerate covariance matrix can be related to a matrix of partial correlations via \citep[p.~130]{lauritzen1996graphical}:
\begin{equation}\label{eq:corr_to_prec_connect}
    \mathbbm{C}orr(\blockG^{[i]}, \blockG^{[j]}|\blockG\backslash\{\blockG^{[i]},\blockG^{[j]}\})=-\frac{\Lambda^{[i,j]}}{\sqrt{\Lambda^{[i,i]}\Lambda^{[j,j]}}}.
\end{equation}
By the definition of $\blockG$ in \eqref{eq:blocking_def_of_G}, it is easy to see that $\mathbbm{C}orr(\blockG^{[i]}, \blockG^{[j]}|\blockG\backslash\{\blockG^{[i]},\blockG^{[j]}\})=0$ whenever $|i-j|>1$; that by symmetry $\Lambda^{[i,i+1]}=\Lambda^{[i+1,i]}$, $(i=1,\dots,m)$; and that $\Lambda^{[i,i]}=\Lambda^{[j,j]}$, $(i,j=1,\dots,m)$, because $\mathbbm{V}ar(\blockG^{[i]}|\blockG\backslash\blockG^{[i]})=(\Lambda^{[i,i]})^{-1}$, $(i=1,\dots,m)$ \citep[p.296]{roberts1997updating}.
In addition, under Assumption \ref{as:gsn_blocks}, the covariance matrix depends only on time and not on the state variable, thus combining this with Assumption \ref{as:equidistant_grid}: $\mathbbm{C}orr(\blockG^{[i]}, \blockG^{[i+1]}|\blockG\backslash\{\blockG^{[i]},\blockG^{[i+1]}\})=:c(\delta_{m,T})$, $(i=1,\dots,m-1)$.
Consequently, $\Lambda$ is a Toeplitz matrix whose non-zero entries are related via $\Lambda^{[i,i+1]}=\Lambda^{[i+1,i]}=-\Lambda^{[i,i]}c(\delta_{m,T})$, $(i=1,\dots,m)$.
The form of matrix $A$ now follows:
\begin{equation}\label{eq:matrix_A}
    A = \left (
\begin{matrix}
0               & c(\delta_{m,T})   & 0                 & \dots             & 0         \\
c(\delta_{m,T}) & 0                 & c(\delta_{m,T})   & \ddots            & \vdots    \\
0               & \ddots            & \ddots            & \ddots            & 0         \\
\vdots          & \ddots            & c(\delta_{m,T})   & 0                 & c(\delta_{m,T}) \\
0               & \dots             & 0                 & c(\delta_{m,T})   & 0
\end{matrix}
\right ).
\end{equation}
The eigenvalues of Toeplitz matrices may be found in closed form \citep{smith1985numerical,kulkarni1999eigenvalues} and in particular, those of matrix $A$ are given by
\begin{equation*}
    -2c(\delta_{m,T})\cos\left(\frac{\pi l}{m+1}\right),\qquad l=1,\dots,m.
\end{equation*}
Depending on the sign of $c(\delta_{m,T})$ the maximal eigenvalue of $A$ is therefore given by:
\begin{equation*}
    \lambda_{\texttt{max}}(A)=
    \begin{cases}
        -2c(\delta_{m,T})\cos\left( \frac{\pi m}{m+1} \right)=2c(\delta_{m,T})\cos\left( \frac{\pi}{m+1} \right),& \mbox{if } c(\delta_{m,T})>0,\\
        -2c(\delta_{m,T})\cos\left( \frac{\pi}{m+1} \right) & \mbox{if } c(\delta_{m,T})<0,
    \end{cases}
\end{equation*}
and the result concerning $\lambda_{\texttt{max}}(A)$ follows. The remaining statements follow as well by substituting the expression for $\lambda_{\texttt{max}}(A)$ into Lemma \ref{lem:general_L2_conv}. 
\end{proof}

\begin{proof}[Proof of Corollary \ref{cor:bm_L2_conv}]
By Theorem \ref{thm:general_explicit_L2_conv}, only $c(\delta_{m,T})$ needs to be computed. This follows from standard properties of Brownian motion and bridges:
\begin{equation}\label{eq:proof_partial_correlation_bm}
    c(\delta_{m,T})= \frac{\mathbbm{C}ov\big(X_{\delta},X_{2\delta}\big|X_0,X_{3\delta}\big)}{\sqrt{\mathbbm{V}ar(X_{\delta}|X_0,X_{3\delta})\mathbbm{V}ar(X_{2\delta}|X_0,X_{3\delta})}}=\frac{\frac{1}{3}\delta\sigma^2}{\sqrt{\left(\frac{2}{3}\delta\sigma^2\right)^2}}=\frac{1}{2}.
\end{equation}
The asymptotic behaviour of $\rho_{m,T}$ follows immediately from Taylor expansion of $\cos^2(x)$ around $0$. For the asymptotic behaviour of $\rho_{\texttt{rand}}$, notice that by Taylor expansions of $\cos(x)$ around $0$, $\log(1-x)$ around $0$, and $\exp(x)$ around $0$ respectively:
\begin{equation*}
    \begin{split}
        \rho_{\texttt{rand}}  &= 
                    \exp\bigg\{
                        m\log\bigg[
                            m^{-1}\bigg\{m-1+\cos\bigg(\frac{\pi}{m+1}\bigg)\bigg\}
                            \bigg]
                        \bigg\} \\
                        &= 
                    \exp\bigg\{
                        m\log\bigg[
                            1-\frac{1}{2m}\bigg(\frac{\pi}{m+1}\bigg)^2 + \mathcal{O}(m^{-5})\bigg\}
                            \bigg]
                        \bigg\} \\
                        &=
                    \exp\bigg\{
                        -\frac{1}{2}\bigg(\frac{\pi}{m+1}\bigg)^2 + \mathcal{O}(m^{-4})
                        \bigg\} \\
                        &=
                    1-\frac{1}{2}\bigg(\frac{\pi}{m+1}\bigg)^2 + \mathcal{O}(m^{-4}).
    \end{split}
\end{equation*} 
\end{proof}

\begin{proof}[Proof of Corollary \ref{cor:ou_L2_conv}]
For the Ornstein--Uhlenbeck process we have:
\begin{equation*}
\begin{split}
    \mathbbm{C}ov\left[\left(\left.\begin{matrix}Y_s\\ Y_t\end{matrix}\right)\right|Y_0,Y_T\right]&=\frac{\sigma^2}{\theta}
\left(
\begin{matrix}
e^{-\theta s}\sinh(\theta s)& e^{-\theta t}\sinh(\theta s) \\
e^{-\theta t}\sinh(\theta s)& e^{-\theta t}\sinh(\theta t) \\
\end{matrix}
\right)\\
&\qquad-\frac{\sigma^2}{\theta}\left(
\begin{matrix}
e^{-\theta T}\frac{\sinh^2(\theta s)}{\sinh(\theta T)} & e^{-\theta T}\frac{\sinh(\theta s)\sinh(\theta t)}{\sinh(\theta T)}\\
e^{-\theta T}\frac{\sinh(\theta s)\sinh(\theta t)}{\sinh(\theta T)}& e^{-\theta T}\frac{\sinh^2(\theta t)}{\sinh(\theta T)} \\
\end{matrix}
\right),\quad 0<s<t<T,
\end{split}
\end{equation*}
which for the particular choice of $(s,t,T)=(\delta,2\delta,3\delta)$ simplifies to:
\begin{equation}\label{eq:covariance_ou}
    \mathbbm{C}ov\left[\left(\left.\begin{matrix}Y_{\delta}\\ Y_{2\delta}\end{matrix}\right)\right|Y_0,Y_{3\delta}\right]=\frac{\sigma^2}{\theta}\frac{\sinh^2(\theta \delta)}{\sinh(3\theta\delta)}
\left(
\begin{matrix}
2\cosh(\theta \delta)& 1 \\
1& 2\cosh(\theta \delta) \\
\end{matrix}
\right).
\end{equation}

It now follows from direct substitution of the relevant terms of \eqref{eq:covariance_ou} into the definition of the partial correlation that:
\begin{equation}\label{eq:OU_corr}
    c(\delta_{m,T})=\frac{\mathbbm{C}ov\big(X_{\delta},X_{2\delta}\big|X_0,X_{3\delta}\big)}{\sqrt{\mathbbm{V}ar(X_{\delta}|X_0,X_{3\delta})\mathbbm{V}ar(X_{2\delta}|X_0,X_{3\delta})}} = \frac{1}{2\cosh(\theta \delta)}.
\end{equation}
The remaining steps are analogous to the proof of Corollary \ref{cor:bm_L2_conv}. 
\end{proof}

\begin{proof}[Proof of Theorem \ref{thm:number_of_needed_steps}]
  The result follows easily from \eqref{eq:relaxation_time_vs_spectral_radius} and Corollaries \ref{cor:bm_L2_conv} and \ref{cor:ou_L2_conv}. For example if $\law$ denotes the law of a scaled Brownian motion then
  \[
  \mathcal{T}(m) = -\frac{1}{\log \rho_{m,T}} = -\left[\log\left(1-\frac{\pi^2}{(m+1)^2} + \mathcal{O}(m^{-4})\right)\right]^{-1} = -\left[\frac{\pi^2}{(m+1)^2} + \mathcal{O}(m^{-4})\right]^{-1} = \mathcal{O}(m^2)
  \]
  as $m\to\infty$, as required. Related expressions for the relaxation time of the Ornstein--Uhlenbeck process, and for the relaxation time relating to $\rho_{\texttt{rand}}$, follow similarly. 
\end{proof}

\begin{proof}[Proof of Theorem \ref{thm:main_thm_comp_cost}]
This follows from Theorem \ref{thm:number_of_needed_steps} and  \eqref{eq:blocking_cost_template1}. We minimize the cost of blocking $C_{\texttt{blocking}}(T,m)$ over the remaining hyperparameter $m$ and derive its final form as a function of $T$.
\begin{enumerate}
\item If $\law$ is the law of the Ornstein--Uhlenbeck process, then the restriction $T=\ctc m$ from Theorem \ref{thm:number_of_needed_steps} already constraints the choice of $m$. Additionally, under the choice of small enough $c_1$: $\delta_{m,T}<\cte $ for all $m$ and $T$, and thus, $C_{\texttt{blocking}}(T,m)\sim\mathcal{T}(m)T=\mathcal{O}(T)$.
\item On the other hand, if $\law$ is the law of the scaled Brownian motion, then the fastest growing contribution is that from the exponential term in \eqref{eq:blocking_cost_template1} and in order to annul it, we should take $m=\ctf T$. 
\end{enumerate} 
\end{proof}


\section*{Acknowledgements}
MM was supported as a doctoral student at the Department of Statistics, University of Warwick under Engineering and Physical Sciences Research Council (EPSRC) grant EP/L016710/1, and his work was concluded while being supported by the Max Planck Institute for Mathematics in the Sciences, Leipzig. PJ, MP, and GOR were supported by The Alan Turing Institute under the EPSRC grant EP/N510129/1. GOR was additionally supported under the EPSRC grants EP/K034154/1, EP/K014463/1, and EP/R018561/1. We are grateful to two anonymous referees whose constructive comments have led to substantial improvements over a previous version of this paper.

\section*{Declarations}
\noindent {\bf Competing interests.} The authors declare that they have no conflict of interest.

\noindent {\bf Availability of data and materials.} Code sufficient for reproducing the results in this article are available in the Github repository \url{https://github.com/mmider/blocking}. 


\bibliographystyle{agsm}
\bibliography{paper-ref}

@article{rogers1985smooth,
  title={Smooth Transition Densities for One-Dimensional Diffusions},
  author={Rogers, LCG},
  journal={Bulletin of the London Mathematical Society},
  volume={17},
  number={2},
  pages={157--161},
  year={1985},
  publisher={Wiley Online Library}
}

@article{shephard1997likelihood,
  title={Likelihood analysis of non-{Gaussian} measurement time series},
  author={Shephard, Neil and Pitt, Michael K},
  journal={Biometrika},
  volume={84},
  number={3},
  pages={653--667},
  year={1997},
  publisher={Biometrika Trust}
}

@article{stramer2007bayesian,
  title={On {Bayesian} analysis of nonlinear continuous-time autoregression models},
  author={Stramer, Osnat and Roberts, Gareth O},
  journal={Journal of Time Series Analysis},
  volume=28,
  number=5,
  pages={744--762},
  year=2007,
  publisher={Wiley Online Library}
}

@article{beskos2005exact,
  title={Exact simulation of diffusions},
  author={Beskos, Alexandros and Roberts, Gareth O},
  journal={The Annals of Applied Probability},
  volume={15},
  number={4},
  pages={2422--2444},
  year={2005},
  publisher={Institute of Mathematical Statistics}
}

@article{beskos2006retrospective,
  title={Retrospective exact simulation of diffusion sample paths with applications},
  author={Beskos, Alexandros and Papaspiliopoulos, Omiros and Roberts, Gareth O},
  journal={Bernoulli},
  pages={1077--1098},
  volume={12},
  number={6},
  year={2006},
  publisher={JSTOR}
}

@article{beskos2008factorisation,
  title={A factorisation of diffusion measure and finite sample path constructions},
  author={Beskos, Alexandros and Papaspiliopoulos, Omiros and Roberts, Gareth O},
  journal={Methodology and Computing in Applied Probability},
  volume={10},
  number={1},
  pages={85--104},
  year={2008},
  publisher={Springer}}

@article{b:pjr16,
    author      = {Pollock, M. and Johansen, A.M. and Roberts, G.O.},
    title       = {{On the exact and $\varepsilon$-strong simulation of (jump) diffusions}},
    journal     = {Bernoulli},
    volume      = {22},
    number      = {2},
    pages       = {794--856},
    year        = {2016}}

@article{amit1991rates,
  title={On rates of convergence of stochastic relaxation for {Gaussian} and non-{Gaussian} distributions},
  author={Amit, Yali},
  journal={Journal of Multivariate Analysis},
  volume={38},
  number={1},
  pages={82--99},
  year={1991},
  publisher={Elsevier}
}

@article{roberts1997updating,
  title={Updating schemes, correlation structure, blocking and parameterization for the {Gibbs} sampler},
  author={Roberts, Gareth O and Sahu, Sujit K},
  journal={Journal of the Royal Statistical Society: Series B (Statistical Methodology)},
  volume={59},
  number={2},
  pages={291--317},
  year={1997},
  publisher={Wiley Online Library}
}

@book{lauritzen1996graphical,
  title={Graphical models},
  author={Lauritzen, Steffen L},
  volume={17},
  year={1996},
  publisher={Clarendon Press}
}

@book{smith1985numerical,
  title={Numerical solution of partial differential equations: finite difference methods},
  author={Smith, Gordon D},
  year={1985},
  publisher={Oxford University Press}
}

@article{kulkarni1999eigenvalues,
  title={Eigenvalues of tridiagonal pseudo-{Toeplitz} matrices},
  author={Kulkarni, Devadatta and Schmidt, Darrell and Tsui, Sze Kai},
  journal={Linear Algebra and its Applications},
  volume={297},
  number={1-3},
  pages={63--80},
  year={1999},
  publisher={Elsevier}
}

@article{kalogeropoulos2007likelihood,
  title={Likelihood-based inference for a class of multivariate diffusions with unobserved paths},
  author={Kalogeropoulos, Konstantinos},
  journal={Journal of Statistical Planning and Inference},
  volume={137},
  number={10},
  pages={3092--3102},
  year={2007},
  publisher={Elsevier}
}

@article{kalogeropoulos2010inference,
  title={Inference for stochastic volatility models using time change transformations},
  author={Kalogeropoulos, Konstantinos and Roberts, Gareth O and Dellaportas, Petros},
  journal={The Annals of Statistics},
  volume=38,
  number=2,
  pages={784--807},
  year=2010,
  publisher={Institute of Mathematical Statistics}
}

@misc{chib2004likelihood,
  title={Likelihood based inference for diffusion driven models},
  author={Chib, Siddhartha and Pitt, Michael K and Shephard, Neil},
  journal={Oxford Financial Research Centre},
  year={2004}
}

@article{golightly_bayesian:2008,
	title = {Bayesian inference for nonlinear multivariate diffusion models observed with error},
	volume = {52},
	number = {3},
	journal = {Computational Statistics \& Data Analysis},
	author = {Golightly, Andrew and Wilkinson, Darren J},
	year = {2008},
	pages = {1674--1693},
}

@article{van2018bayesian,
  title={Bayesian estimation of incompletely observed diffusions},
  author={van der Meulen, Frank and Schauer, Moritz},
  journal={Stochastics},
  volume={90},
  number={5},
  pages={641--662},
  year={2018},
  publisher={Taylor \& Francis}
}

@article{roberts2001inference,
  title={On inference for partially observed nonlinear diffusion models using the Metropolis--Hastings algorithm},
  author={Roberts, Gareth O and Stramer, Osnat},
  journal={Biometrika},
  volume={88},
  number={3},
  pages={603--621},
  year={2001},
  publisher={Oxford University Press}
}

@article{boys2008bayesian,
  title={Bayesian inference for a discretely observed stochastic kinetic model},
  author={Boys, Richard J and Wilkinson, Darren J and Kirkwood, Thomas B L},
  journal={Statistics and Computing},
  volume={18},
  number={2},
  pages={125--135},
  year={2008},
  publisher={Springer}
}

@article{lansky2008review,
  title={A review of the methods for signal estimation in stochastic diffusion leaky integrate-and-fire neuronal models},
  author={Lansky, Petr and Ditlevsen, Susanne},
  journal={Biological cybernetics},
  volume={99},
  number={4-5},
  pages={253},
  year={2008},
  publisher={Springer}
}

@book{karatzas1998methods,
  title={Methods of mathematical finance},
  author={Karatzas, Ioannis and Shreve, Steven E},
  volume={39},
  year={1998},
  publisher={Springer}
}

@book{kloeden1992numerical,
  title={Numerical {Solution} of {Stochastic} {Differential} {Equations}},
  author={Kloeden, P E and Platen, E},
  isbn={9783540540625},
  lccn={92015916},
  volume={23},
  year={2013},
  publisher={Springer Science \& Business Media}
}

@article{schauer2017guided,
  title={Guided proposals for simulating multi-dimensional diffusion bridges},
  author={Schauer, Moritz and Van Der Meulen, Frank and Van Zanten, Harry and others},
  journal={Bernoulli},
  volume={23},
  number={4A},
  pages={2917--2950},
  year={2017},
  publisher={Bernoulli Society for Mathematical Statistics and Probability}
}

@article{delyon2006simulation,
  title={Simulation of conditioned diffusion and application to parameter estimation},
  author={Delyon, Bernard and Hu, Ying},
  journal={Stochastic Processes and their Applications},
  volume={116},
  number={11},
  pages={1660--1675},
  year={2006},
  publisher={Elsevier}
}

@article{bladt2014simple,
  title={Simple simulation of diffusion bridges with application to likelihood inference for diffusions},
  author={Bladt, Mogens and S{\o}rensen, Michael and others},
  journal={Bernoulli},
  volume={20},
  number={2},
  pages={645--675},
  year={2014},
  publisher={Bernoulli Society for Mathematical Statistics and Probability}
}

@article{arnaudon2020diffusion,
  title={Diffusion bridges for stochastic {Hamiltonian} systems with applications to shape analysis},
  author={Arnaudon, Alexis and van der Meulen, Frank and Schauer, Moritz and Sommer, Stefan},
  journal={arXiv preprint arXiv:2002.00885},
  year={2020}
}

@book{oksendal2003stochastic,
  title={Stochastic differential equations},
  author={{\O}ksendal, Bernt},
  booktitle={Stochastic differential equations},
  pages={65--84},
  year={2003},
  publisher={Springer}
}

@article{hairer2011sampling,
  title={Sampling conditioned hypoelliptic diffusions},
  author={Hairer, Martin and Stuart, Andrew M and Voss, Jochen and others},
  journal={The Annals of Applied Probability},
  volume={21},
  number={2},
  pages={669--698},
  year={2011},
  publisher={Institute of Mathematical Statistics}
}

@article{durham2002numerical,
  title={Numerical techniques for maximum likelihood estimation of continuous-time diffusion processes},
  author={Durham, Garland B and Gallant, A Ronald},
  journal={Journal of Business \& Economic Statistics},
  volume={20},
  number={3},
  pages={297--338},
  year={2002},
  publisher={Taylor \& Francis}
}

@article{freidlin1993diffusion,
  title={Diffusion processes on graphs and the averaging principle},
  author={Freidlin, Mark I and Wentzell, Alexander D},
  journal={The Annals of probability},
  pages={2215--2245},
  year={1993},
  publisher={JSTOR}
}

@book{levin2017markov,
  title={Markov chains and mixing times},
  author={Levin, David A and Peres, Yuval},
  volume={107},
  year={2017},
  publisher={American Mathematical Soc.}
}

@article{dacunha1986estimation,
  title={Estimation of the coefficients of a diffusion from discrete observations},
  author={Dacunha-Castelle, Didier and Florens-Zmirou, Danielle},
  journal={Stochastics: An International Journal of Probability and Stochastic Processes},
  volume={19},
  number={4},
  pages={263--284},
  year={1986},
  publisher={Taylor \& Francis}
}

@book{gelman2013bayesian,
  title={Bayesian data analysis},
  author={Gelman, Andrew and Carlin, John B and Stern, Hal S and Dunson, David B and Vehtari, Aki and Rubin, Donald B},
  year={2013},
  publisher={CRC press}
}

@article{roberts2001markov,
  title={Markov chains and de-initializing processes},
  author={Roberts, Gareth O and Rosenthal, Jeffrey S},
  journal={Scandinavian Journal of Statistics},
  volume={28},
  number={3},
  pages={489--504},
  year={2001},
  publisher={Wiley Online Library}
}

@Article{lat:etal:2011,
	author =	{K. {\L atuszy\'nski} and I. Kosmidis and O. Papaspiliopoulos and G. O. Roberts},
	title =	{Simulating events of unknown probabilities via reverse time martingales},
	year =	2011,
	journal =	{Random structures and algorithms},
	volume =	38,
	number =	4,
	pages =	{441--452}
}

@article{pit:she:1999:JTSA,
  title={Analytic convergence rates and parameterization issues for the {Gibbs} sampler applied to state space models},
  author={M. K. Pitt and N. Shephard},
  journal={Journal of Time Series Analysis},
  volume={20},
  number={1},
  pages={63--85},
  year={1999}
}

\end{document}